\documentclass[letterpaper, 10 pt, conference]{ieeeconf}  

\IEEEoverridecommandlockouts                              
\usepackage{amsmath,amsfonts,amssymb,color}
\usepackage{graphicx}
\usepackage{psfrag}
\usepackage{algorithm}
\usepackage{algpseudocode}
\usepackage{array}
\usepackage{cite}
\usepackage{footmisc}
\usepackage{mathtools}
\usepackage{enumerate}
\usepackage{soul,color}

\newtheorem{theorem}{Theorem}
\newtheorem{corollary}{Corollary}
\newtheorem{lemma}{Lemma}
\newtheorem{definition}{Definition}
\newtheorem{proposition}{Proposition}
\newtheorem{example}{Example}
\newtheorem{assumption}{Assumption}

\renewcommand{\theenumi}{(\alph{enumi}}
\newtheorem{Remark}{Remark}

\DeclareMathOperator{\diag}{diag}  

\DeclareMathOperator{\rank}{rank}

\title{\LARGE \bf
Byzantine-Resilient Distributed Optimization of \\ Multi-Dimensional Functions
}

\author{Kananart Kuwaranancharoen, Lei Xin, Shreyas Sundaram
\thanks{This research was supported by NSF CAREER award 1653648. The authors are with the School of Electrical and Computer Engineering at Purdue University.  Email: {\tt\{kkuwaran,lxin,sundara2\}@purdue.edu}. Both of the first two authors contributed equally to this paper.}
}

\begin{document}
    \title{\LARGE \bf Learning the Dynamics of Autonomous Linear Systems From Multiple Trajectories}
  \author{~Lei~Xin, George Chiu, Shreyas Sundaram 
\thanks{This research was supported by USDA grant 2018-67007-28439.  This work represents the opinions of the authors and not the USDA or NIFA. Lei Xin and Shreyas Sundaram are with the Elmore Family School of Electrical and Computer Engineering, Purdue University. George Chiu is with the School of Mechanical Engineering, Purdue University. E-mails: {\tt\{lxin, gchiu, sundara2\}@purdue.edu}.
}
}

\maketitle

\begin{abstract}
We consider the problem of learning the dynamics of autonomous linear systems (i.e., systems that are not affected by external control inputs) from observations of multiple trajectories of those systems, with finite sample guarantees. Existing results on learning rate and consistency of autonomous linear system identification rely on observations of steady state behaviors from a single long trajectory, and are not applicable to unstable systems. In contrast, we consider the scenario of learning system dynamics based on multiple short trajectories, where there are no easily observed steady state behaviors. We provide a finite sample analysis, which shows that the dynamics can be learned at a rate $\mathcal{O}(\frac{1}{\sqrt{N}})$ for both stable and unstable systems, where $N$ is the number of trajectories, when the initial state of the system has zero mean (which is a common assumption in the existing literature). We further generalize our result to the case where the initial state has non-zero mean. We show that one can adjust the length of the trajectories to achieve a learning rate of $\mathcal{O}(\sqrt{\frac{\log{N}}{N})}$ for strictly stable systems and a learning rate of $\mathcal{O}(\frac{(\log{N})^d}{\sqrt{N}})$ for marginally stable systems, where $d$ is some constant. 
\end{abstract}


\section{Introduction} \label{sec: introduction}
The system identification problem is to learn the parameters of a dynamical system, given the measurements of the inputs and outputs. While classical system identification techniques focused primarily on achieving asymptotic consistency \cite{bauer1999consistency,jansson1998consistency,knudsen2001consistency}, recent efforts have sought to characterize the number of samples needed to achieve a desired level of accuracy in the learned model. In particular, non-asymptotic analysis of system identification based on a single trajectory has been studied extensively over the past few years \cite{simchowitz2018learning,oymak2018non,simchowitz2019learning,sarkar2019finite,faradonbeh2018finite,sarkar2019near}.

In practice, performing system identification using a single trajectory could be problematic for many applications. For example, having the system run for a long time could incur risks when the system is unstable. Furthermore, only historical snippets of data about the system may be available, without the ability to easily observe long-run behavior. Additionally, in settings where one has the ability to restart the system or have several copies of the system running in parallel, one may obtain \textit{multiple} trajectories generated by the system dynamics  \cite{xing2020linear}. The paper \cite{dean2019sample} studies the sample complexity of identifying a system whose state is fully measured using only the final data points from multiple trajectories. Using a similar setup, the paper \cite{fattahi2018data} explores the benefits of adding an $\ell_{1}$ regularizer. The paper \cite{sun2020finite} studies the sample complexity of partially-measured system identification by including nuclear norm regularization, again only using the final samples from each trajectory. For partially-measured systems, the paper \cite{zheng2020non} allows for more efficient use of data. As mentioned in \cite{zheng2020non}, compared to the single trajectory setup, the multiple trajectories setup usually allows for more direct application of concentration inequalities due to the assumption of independence over multiple trajectories.

In addition to the potential lack of single long trajectories, in many settings we may not be able to actually apply inputs to the system in order to perform system identification; this could be due to the costs of applying inputs, or due to the fact that we are simply observing an autonomous system that we cannot control. The uncontrolled system may also be serving as a subsystem connected to the main system that one wants to control, and having a better model of the subsystem could be useful in controlling the main system. For partially-measured systems, the characterization of finite sample error of purely stochastic systems (systems that are entirely driven by unmeasurable noise) is more challenging as indicated in \cite{tsiamis2019finite}. There, the goal is to estimate the system matrices as well as the steady state Kalman filter gain of the corresponding system. The paper \cite{tsiamis2019finite} shows that classical stochastic system identification algorithm can achieve a learning rate of $\mathcal{O}(\frac{1}{\sqrt{\bar{N}}})$ (up to logarithmic factors) for both strictly stable and marginally stable systems, where $\bar{N}$ denotes the number of samples in a single trajectory.

In this paper, we are motivated by the challenge of system identification for partially-measured and autonomous stochastic linear systems (with no controlled inputs) as in \cite{tsiamis2019finite}, but for the case where a single long-run trajectory is unavailable. Existing results on consistency and learning rate of stochastic system identification algorithms (including \cite{tsiamis2019finite}) typically convert the original system to a statistically equivalent form of the Kalman filter that is assumed to have reached steady state \cite{bauer1999consistency,deistler1995consistency,tsiamis2019finite}. The analysis is then performed by assuming that the covariance matrix of the initial state of the system is the same as the steady state Kalman filter error covariance matrix, which simplifies the analysis. We note, however, that this assumption is invalid when one has no long run observation of the system trajectory, since it is in general unclear how long one should wait until the Kalman filter “converges” (even if it converges exponentially fast) for an unknown system. Furthermore, the available short trajectories may not be long enough to guarantee that the underlying filter has converged. Consequently, the single trajectory-based results cannot be directly applied to the multiple (short) trajectories case. Our goal in this paper is to estimate the system matrices (up to similarity transformations) using only multiple trajectories of transient responses of a partially-measured system that is entirely driven by noise.


\subsection{Contributions} Our work is inspired by recent work on stochastic system identification (with a single long trajectory) \cite{tsiamis2019finite}, and system identification with multiple trajectories (but with controlled inputs) \cite{zheng2020non}, and extends them in the following ways.

\begin{itemize}
  \item We provide results on the sample complexity of learning the dynamics of an autonomous stochastic linear system using multiple trajectories, without assuming that the associated Kalman filter has reached steady state (i.e., the initial states can have arbitrary covariance matrix). Compared to \cite{zheng2020non} and \cite{tsiamis2019finite}, our results neither rely on controlled inputs, nor rely on observations of steady state behaviors of the system.
  \item We provide the asymptotic learning rate of the multiple-trajectories-based stochastic system identification algorithm. If $N$ is the number of trajectories, we prove a learning rate of $\mathcal{O}(\frac{1}{\sqrt{N}})$ when the initial state in each trajectory has zero mean (which is a common assumption in the existing literature). This rate is consistent with \cite{zheng2020non} and \cite{tsiamis2019finite} (up to logarithmic factors). We further generalize the result to the case when the initial state in each trajectory has non-zero mean. In such case, we show that we can adjust the length of the regressor to achieve a learning rate of $\mathcal{O}(\sqrt{\frac{\log{N}}{N})}$ for strictly stable systems and a learning rate of $\mathcal{O}(\frac{(\log{N})^d}{\sqrt{N}})$ for marginally stable systems, where $d$ is some constant.
\end{itemize}

\section{Mathematical notation and terminology} \label{sec: notation and terminology}
Let $\mathbb{R}$ and $\mathbb{N}$  denote the sets of real numbers and natural numbers, respectively. Let $\sigma_n(\cdot)$ and $\sigma_{min}(\cdot)$ be the $n$-th largest and smallest singular value, respectively, of a symmetric matrix. Similarly, let $\lambda_{min}(\cdot)$ be the smallest eigenvalue of a symmetric matrix. We use $*$ to denote the conjugate transpose of a given matrix. The spectral radius of a given matrix is denoted as $\rho(\cdot)$. A square matrix $A$ is called strictly stable if $\rho(A)<1$, marginally stable if $\rho(A)\leq1$, and unstable if $\rho(A)>1$. We use $\|\cdot\|$ to denote the spectral norm of a given matrix. Vectors are treated as column vectors. A Gaussian distributed random vector is denoted as $u\sim \mathcal{N}(\mu,\Sigma)$, where $\mu$ is the mean and $\Sigma$ is the covariance matrix. We use $I$ to denote the identity matrix. We use $\mathbb{E}$ to denote the expectation. The symbol $\prod_{t=i}^{j}A_{t}$ is used to denote the matrix product, $A_{i}A_{i+1}\cdots A_{j}$. The symbol $\dagger$ is used to denote the pseudoinverse.

\section{Problem formulation} \label{sec: problem formulation}

Consider a discrete time linear time-invariant system with no user specified inputs:
\begin{equation}
    \begin{aligned} 
   x_{k+1}&=Ax_{k}+w_{k},&
   y_{k}&=Cx_{k}+v_{k}  \label{originalStateSpace},
\end{aligned}
\end{equation}
where $x_{k}\in \mathbb{R}^{n}$, $y_{k}\in \mathbb{R}^{m}$, $w_{k}\in \mathbb{R}^{n}$, $v_{k}\in \mathbb{R}^{m}$, $A\in \mathbb{R}^{n \times n}$ and $C\in \mathbb{R}^{m\times n}$. The noise terms $w_{k}$ and $v_{k}$ are assumed to be i.i.d Gaussian, i.e., $w_{k} \sim \mathcal{N}(0,Q)$, $v_{k} \sim \mathcal{N}(0,R)$. The initial state is also assumed to be independent of $w_{k}$ and $v_{k}$, and is distributed as $x_{0} \sim \mathcal{N}(\mu,\Sigma_{0})$. In addition, whether $\mu$ is zero or non-zero is assumed to be known. If $\mu$ is non-zero, the system matrix $A$ is assumed to be strictly stable or marginally stable. The system order $n$ is also assumed to be known. We refer to the above system as an autonomous stochastic linear system. We will make the following assumption. 
\begin{assumption} \label{asm: standard}
The output covariance matrix $R$ is positive definite. The pair $(A,C)$ is observable and $(A,Q^{\frac{1}{2}})$ is controllable.
\end{assumption}

Under the above assumption, the \textit{Kalman Filter} associated with system \eqref{originalStateSpace} is a system of the form
\begin{equation}
    \begin{aligned} 
   \hat{x}_{k+1}&=A\hat{x}_{k}+K_{k}e_{k},&
   y_{k}&=C\hat{x}_{k}+e_{k} \label{KalmanFilterSystem},
\end{aligned}
\end{equation}
where $\hat{x}_{k}$ is an estimate of state $x_{k}$, with the initial estimate being the mean of the initial state in system \eqref{originalStateSpace}, i.e., $\hat{x}_{0}=\mu$. The sequence of matrices $K_{k}\in \mathbb{R}^{n\times m}$ is called the Kalman gain, given by
\begin{equation} \label{gain}
  K_{k}=AP_{k}C^{*}(CP_{k}C^{*}+R)^{-1},
\end{equation}
where the estimation error covariance $P_{k}\in \mathbb{R}^{n\times n}$ is updated based on the Riccati equation
\begin{equation*} 
  P_{k+1}=AP_{k}A^{*}+Q-AP_{k}C^{*}(CP_{k}C^{*}+R)^{-1}CP_{k}A^{*},
\end{equation*}
with $P_{0}=\Sigma_{0}$. Finally, $e_{k}=y_{k}-C\hat{x}_{k}$ are independent zero mean Gaussian innovations with covariance matrix given by 
\begin{align}  \label{innovation covariance}
  \Bar{R}_{k}=CP_{k}C^{*}+R.
\end{align}

Since the outputs of system \eqref{originalStateSpace} and system \eqref{KalmanFilterSystem} have identical statistical properties \cite{van2012subspace}, we will perform analysis on system \eqref{KalmanFilterSystem}. The subspace identification problem for stochastic systems that we tackle in this paper is to identify the system matrices $(A,C)$ up to a similarity transformation, given {\it multiple trajectories} of outputs of the system \eqref{originalStateSpace}. As a byproduct, we will also simultaneously learn the Kalman filter gain $K_{k}$ of the corresponding system, at some time step $k$. In particular, we are interested in the quality of the estimates of $(A,C)$ given a finite number of samples. 
\section{Subspace identification technique} \label{sec: algorithm}
Here we describe a variant of classical subspace identification algorithm \cite{van2012subspace} to estimate $(A,C)$ matrices (up to a similarity transformation). We will first establish some definitions.

Suppose that we have access to $N$ independent output trajectories of system \eqref{originalStateSpace}, each of some length $T\in \mathbb{N}$, and each obtained right after restarting the system from an initial state $x_{0}\sim \mathcal{N}(\mu,\Sigma_{0})$. We denote the data from these trajectories as $\{y^{i}_{k}:1 \leq i \leq N,0 \leq k \leq T-1\}$, where the superscript denotes the trajectory index and the subscript denotes the time index. Let $p+f=T$, where $p,f$ are design parameters that satisfy $p,f>n$, where $n$ is the order of the system. We split the output samples from each output trajectory $i$ into past and future outputs with respect to $p$, and denote the past output and future output vectors for trajectory $i$ as:  
\begin{equation} \label{past output and furure output}
\begin{aligned} 
&Y_{-}^{i}\triangleq 
\begin{bmatrix}
y_{0}^{i*}&y_{1}^{i*}&\cdots&y_{p-1}^{i*}
\end{bmatrix}^{*},\\
&Y_{+}^{i}\triangleq 
\begin{bmatrix}
y_{p}^{i*}&y_{p+1}^{i*}&\cdots&y_{p+f-1}^{i*}
\end{bmatrix}^{*},
\end{aligned}
\end{equation}
respectively. 
The batch past output and batch future output matrices are formed by stacking all $N$ output trajectories:
\begin{equation}\label{batch output} 
\begin{aligned}
Y_{-}\triangleq 
\begin{bmatrix}
Y_{-}^{1}&Y_{-}^{2}&\cdots&Y_{-}^{N}\\
\end{bmatrix},     
&& 
Y_{+}\triangleq 
\begin{bmatrix}
Y_{+}^{1}&Y_{+}^{2}&\cdots&Y_{+}^{N}\\
\end{bmatrix}.
\end{aligned}
\end{equation}
\noindent The past and future innovations $E_{-}^{i},E_{+}^{i},E_{-},E_{+}$ are defined similarly, using the signals $e_{k}^{i}$ rather than $y_{k}^{i}$.

Let the batch matrix of initial states be $
\hat{X}_{0}\triangleq 
\begin{bmatrix}
\hat{x}_0^{1}&\hat{x}_0^{2}&\cdots&\hat{x}_0^{N}\\
\end{bmatrix}.
$ Define the largest norm of innovation covariance matrices as
$\mathcal{\bar{R}}_{T}\triangleq \max_{t\in0,\dots,T-1}\|\bar{R}_{t}\|,$
where $\bar{R}_{t}$ is defined in \eqref{innovation covariance}. For any $l\geq1$, the extended observability matrix $\mathcal{O}_{l}\in\mathbb{R}^{ml\times n}$ and the reversed extended controllability matrix $\mathcal{K}_{p}\in\mathbb{R}^{n\times mp}$ are defined as:
\begin{equation*}
\begin{aligned} 
&\mathcal{O}_{l}\triangleq 
\begin{bmatrix}
C^{*} & (CA)^{*} & \cdots& (CA^{l-1})^{*}
\end{bmatrix}^{*}, \\
&\mathcal{K}_{p}\triangleq 
\begin{bmatrix}
((A-K_{p-1}C)\cdots(A-K_{1}C)K_{0})^{*}\\
\vdots\\
((A-K_{p-1}C)K_{p-2})^{*}\\
K_{p-1}^{*}\\
\end{bmatrix}^{*}.   
\end{aligned}
\end{equation*}
Define
\begin{equation}
\begin{aligned}
G\triangleq \mathcal{O}_{f} \mathcal{K}_{p}.
\end{aligned} \label{True G}
\end{equation}

Let $K\in \mathbb{R}^{n\times m}$ be the steady state Kalman gain $K=APC^{*}(CPC^{*}+R)^{-1}$, where $P\in \mathbb{R}^{n\times n}$ is the solution to the Riccati equation, $P=APA^{*}+Q-APC^{*}(CPC^{*}+R)^{-1}CPA^{*}.$
From Kalman filtering theory, the matrix $A-KC$ has spectral radius strictly less than 1 \cite{anderson2012optimal}. Denote the reversed extended controllability matrix formed by the steady state Kalman gain $K$ as 
\begin{equation*}
\begin{aligned}
\mathbf{K}_{p}\triangleq 
\begin{bmatrix}
(A-KC)^{p-1}K&(A-KC)^{p-2}K&\cdots&K\\
\end{bmatrix}.     
\end{aligned}
\end{equation*}

We further make the following assumption.
\begin{assumption} \label{asm: controllable}
We have $\rank (\mathcal{K}_{p})=\rank(\mathbf{K}_{p})=n$.
\end{assumption}

The rank condition on $\mathbf{K}_{p}$ is standard, e.g., \cite{knudsen2001consistency, tsiamis2019finite}. The rank condition on $\mathcal{K}_{p}$ is needed to ensure that $G$ has rank $n$, which can be satisfied in practice by choosing $p$ to be relatively large if the rank condition on $\mathbf{K}_{p}$ is satisfied (see Proposition \ref{proposition:Lower bound of least singular value of G} in the Appendix).

Finally, for any integer $a\geq 0$ and $b\geq 2$, define the  block-Toeplitz matrix $\mathcal{T}_{b}^{a}\in\mathbb{R}^{mb\times mb}$ as:
\begin{align*}
\mathcal{T}_b^{a} \triangleq 
\begin{bmatrix}
I_{m}&0&\cdots&0\\
CK_{a}&I_{m}&\cdots&0\\
\vdots&\vdots&\indent&\vdots\\
CA^{b-2}K_{a}&CA^{b-3}K_{a+1}&\cdots&I_{m}\\
\end{bmatrix}.
\end{align*}

\subsection{Linear regression}
The subspace identification technique first uses linear regression to estimate $G$ from \eqref{True G}, which will subsequently form the basis for the recovery of the system parameters.

For any output trajectory $i\in\{1,\cdots,N\}$, by iterating \eqref{KalmanFilterSystem}, the future output matrix $Y_{+}^{i}$ satisfies
\begin{equation}
Y_{+}^{i}=\mathcal{O}_{f}\hat{x}_{p}^{i}+\mathcal{T}_{f}^{p}E_{+}^{i}. \label{YfutureExpression}
\end{equation}
Note that at any time step $k$, the state $\hat{x}_{k}^{i}$ can be expressed from \eqref{KalmanFilterSystem} as
\begin{equation*}
\begin{aligned}
\hat{x}_{k}^{i}&=A\hat{x}^{i}_{k-1}+K_{k-1}(y^{i}_{k-1}-C\hat{x}^{i}_{k-1})\\
&=K_{k-1}y^{i}_{k-1}+(A-K_{k-1}C)\hat{x}^{i}_{k-1}.
\end{aligned}
\end{equation*}

By expanding the above relationship recursively, we have
\begin{equation*}
\begin{aligned}
\hat{x}_{p}^{i}&=K_{p-1}y_{p-1}^{i}+\cdots+
(A-K_{p-1}C)\cdots(A-K_{1}C)K_{0}y_{0}^{i}\\
&+(A-K_{p-1}C)\cdots(A-K_{0}C)\hat{x}_{0}^{i}\\
&=\mathcal{K}_{p}Y_{-}^{i}+(A-K_{p-1}C)\cdots(A-K_{0}C)\hat{x}_{0}^{i}. \label{TVX}
\end{aligned}
\end{equation*}

By substituting the above equality into \eqref{YfutureExpression}, the relationship between the batch output matrices is given by
\begin{equation*} 
Y_{+}=GY_{-}+\mathcal{O}_{f}(A-K_{p-1}C)\cdots(A-K_{0}C)\hat{X}_{0}+\mathcal{T}_{f}^{p}E_{+}.
\end{equation*}
An estimate of $G$ (motivated by the least squares approach) is  
\begin{equation}
\begin{aligned}
\hat{G}=Y_{+}Y_{-}^{*}(Y_{-}Y_{-}^{*})^{-1}.\label{Ghat}
\end{aligned}
\end{equation}
Consequently, the estimation error for matrix $G$ can be expressed as
\begin{equation}
\begin{aligned}
\hat{G}-G=&\mathcal{T}_{f}^{p}E_{+}Y_{-}^{*}(Y_{-}Y_{-}^{*})^{-1}+\\
&\mathcal{O}_{f}(A-K_{p-1}C)\cdots(A-K_{0}C)\hat{X}_{0}Y_{-}^{*}(Y_{-}Y_{-}^{*})^{-1}, \label{error_G}
\end{aligned}
\end{equation}
where the second term can be dropped if $\|\hat{X}_{0}\|=0$, i.e., the initial state of system \eqref{originalStateSpace} has zero mean. When $\|\hat{X}_{0}\|$ is known to be non-zero but the system is marginally stable ($\rho(A)\leq 1$), we can leverage the fact that the norm $\|(A-K_{p-1}C)\cdots(A-K_{0}C)\|$ converges to zero exponentially fast with $p$ (see Proposition \ref{proposition:Convergence of (A-KC)} in the Appendix) by setting $p=c\log{N}$ for some positive constant $c$, to make the second term go to zero asymptotically. The above steps are encapsulated in Algorithm \ref{alg: Rgression Algorithm}.

\begin{algorithm}[H] 
\caption{Linear regression to calculate an estimate $\hat{G}$ of $G$ } \label{Linear Regression} \label{alg: Rgression Algorithm}
\textbf{Input} $N$ output trajectories $\{y^{i}_{k}:1 \leq i \leq N,0 \leq k \leq T-1\}$, parameters $p,f$
\begin{algorithmic}[1]
\State For each output trajectory $i\in\{1,\cdots,N\}$, construct the past output and future output $Y_{-}^{i}, Y_{+}^{i}$ as in \eqref{past output and furure output}.
\State Construct the batch past output and batch future output $Y_{-}, Y_{+}$ as in \eqref{batch output}.
\State Return $\hat{G}=Y_{+}Y_{-}^{*}(Y_{-}Y_{-}^{*})^{-1}$.
\end{algorithmic}
\end{algorithm}
\begin{Remark}Note that the matrix $G$ itself can be used to predict future outputs from past outputs. The value $GY_{-}^{i}$ represents the Kalman prediction of the next $f$ outputs using the past $p$ measurements, assuming the initial state has zero mean. Its role is similar to the Markov parameters that map inputs to outputs in the case when one has measured inputs.
\end{Remark}


\subsection{Balanced realization} 
The following balanced realization algorithm uses a standard Singular Value Decomposition to extract the estimated system matrices $(\hat{A},\hat{C},\hat{K}_{p-1})$ from the estimate $\hat{G}$. 

\begin{algorithm}[H]
\caption{Balanced realization to calculate estimates $(\hat{A},\hat{C},\hat{K}_{p-1})$ of $(A,C,K_{p-1})$ up to a similarity transformation}\label{alg: SVD Algorithm}
\textbf{Input} The estimate $\hat{G}$, parameters $n,m,f$
\begin{algorithmic}[1]
\State Perform the Singular Value Decomposition: 
$\hat{G}=
\begin{bmatrix}
\hat{U}_{1}& \hat{U}_{2}
\end{bmatrix}
\begin{bmatrix}
\hat{\Sigma}_{1}&0\\
0&\hat{\Sigma}_{2}\\
\end{bmatrix}
\begin{bmatrix}
\hat{V}_{1}^{*}\\
\hat{V}_{2}^{*}\\
\end{bmatrix},$ where $\Sigma_{1}\in \mathbb{R}^{n \times n}$ contains the $n$-largest singular values of $\hat{G}$.
\State Compute the estimated observability matrix $\hat{\mathcal{O}}_{f}=\hat{U}_{1}\hat{\Sigma}_{1}^{\frac{1}{2}}$, and let the top $m$ rows of $\hat{\mathcal{O}}_{f}$ be $\hat{C}$.
\State Compute the estimated reversed extended controllability matrix $\hat{\mathcal{K}}_{p}=\hat{\Sigma}_{1}^{\frac{1}{2}}\hat{V}_{1}^{*}$, and let the last $m$ columns of $\hat{\mathcal{K}}_{p}$ be $\hat{K}_{p-1}$.
\State Compute $\hat{A}=(\hat{\mathcal{O}}_{f}^{u})^{\dagger}\hat{\mathcal{O}}_{f}^{l}$, where $\hat{\mathcal{O}}_{f}^{u}$ is the submatrix formed by the top $m(f-1)$ rows of $\hat{\mathcal{O}}_{f}$, and $\hat{\mathcal{O}}_{f}^{l}$ is the submatrix formed by dropping the first $m$ rows of $\hat{\mathcal{O}}_{f}$. \State Return $(\hat{A},\hat{C},\hat{K}_{p-1})$.
\end{algorithmic}
\end{algorithm}

\section{Main results}
In this section, we will present our main results on bounding the estimation error $\|\hat{G}-G\|$ from \eqref{error_G}.  We will show that the term $\|(Y_{-}Y_{-}^{*})^{-1}\|$ decreases with a rate of $\mathcal{O}(\frac{1}{N})$, and then upper bound the growth rate of other terms in \eqref{error_G} separately. Using recent results on the balanced realization algorithm with the adjustments to accommodate the non-steady state Kalman filter, we then show that the estimation error of the system matrices $A,C,K_{p-1}$ will also be bounded when the error $\|\hat{G}-G\|$ is small enough. 

First, denote the covariance matrix of the weighted past innovation $\mathcal{T}_{p}^{0}E_{-}^{i}$, $1\leq i\leq N$, as:
\begin{equation*}
    \begin{aligned} 
   \Sigma_{E}&\triangleq\mathbb{E}[\mathcal{T}_{p}^{0}E_{-}^{i}E_{-}^{i*}\mathcal{T}_{p}^{0*}]=\mathcal{T}_{p}^{0} \diag(\bar{R}_{0},\cdots,\bar{R}_{p-1})\mathcal{T}_{p}^{0*}.
\end{aligned}
\end{equation*}
Let $\sigma_{E}\triangleq\sigma_{min}(\Sigma_{E})$. We first show that the weighted innovation covariance matrix $\Sigma_{E}$ is strictly positive definite. The proof is similar to \cite[Lemma~2]{tsiamis2019finite}.

\begin{proposition} 
Let $\sigma_{E}\triangleq\sigma_{min}(\Sigma_{E})$. We have $\sigma_{E}\geq\sigma_{min}(R)>0$.
\end{proposition}
\begin{proof}
For any output trajectory $i$, its corresponding weighted past innovation $\mathcal{T}_{p}^{0}E_{-}^{i}$ can be written as
\begin{equation*}
\begin{aligned}
\mathcal{T}_{p}^{0}E_{-}^{i}=Y_{-}^{i}-\mathcal{O}_{p}\hat{x}_0^{i}=\mathcal{O}_{p}(x_{0}^{i}-\hat{x}_0^{i})+\mathbf{T}W_{-}^{i}+V_{-}^{i},
\end{aligned}
\end{equation*}
where $W_{-}^{i}$ and $V_{-}^{i}$ are the process and output noises respectively in system \eqref{originalStateSpace}, and are defined similarly as $Y_{-}^{i}$. Matrix $\mathbf{T}$ is a block-Toeplitz matrix which accounts for the weight of the process noise in system \eqref{originalStateSpace}, and its explicit form is omitted in the interest of space. Since $x_0^{i}-\hat{x}_{0}^i$, $V_{-}^{i}$ and $W_{-}^{i}$ are independent, we have
\begin{equation*}
\begin{aligned}
\Sigma_E&=\mathbb{E}[\mathcal{T}_{p}^{0}E_{-}^{i}E_{-}^{i*}\mathcal{T}_{p}^{0*}]\succeq \mathbb{E}[V_{-}^{i}V_{-}^{i*}]=\diag(R,\cdots,R).
\end{aligned}
\end{equation*}
Hence, we have $\sigma_E\geq\sigma_{min}(R)>0$, where the second inequality comes from Assumption \ref{asm: standard}.
\end{proof}

Next we will show that the term $\|(Y_{-}Y_{-}^{*})^{-1}\|$ is decreasing with a rate of $\mathcal{O}(\frac{1}{N})$. Since $Y_{-}=\mathcal{O}_{p}\hat{X}_{0}+\mathcal{T}_{p}^{0}E_{-}$, the explicit form of $Y_{-}Y_{-}^{*}$ is 
\begin{equation}
\begin{aligned}
Y_{-}Y_{-}^{*}&=\mathcal{O}_{p}\hat{X}_{0}\hat{X}_{0}^{*}\mathcal{O}_{p}^{*}+\mathcal{T}_{p}^{0}E_{-}E_{-}^{*}\mathcal{T}_{p}^{0*}\\
&+\mathcal{O}_{p}\hat{X}_{0}E_{-}^{*}\mathcal{T}_{p}^{0*}+\mathcal{T}_{p}^{0}E_{-}\hat{X}_{0}^{*}\mathcal{O}_{p}^{*}, \label{YY^{*}}
\end{aligned}
\end{equation}
and we will bound these terms separately.

We will rely on the following lemma from \cite[Lemma~2]{dean2019sample}, which provides a non-asymptotic lower bound of a standard Wishart matrix.
\begin{lemma} Let $u_{i}\sim \mathcal{N}(0,I_{mp})$, $i=1,\ldots,N$ be i.i.d random vectors. For any fixed $\delta > 0$, we have
\begin{equation*}
   \sqrt{\lambda_{min}(\sum_{i=1}^{N}u_{i}u_{i}^{*}})\geq \sqrt{N}-\sqrt{mp}-\sqrt{2\log{\frac{1}{\delta}}}
\end{equation*}
\label{lemma:Lower bound of unit variance gaussian}
with probability at least $1-\delta$.
\end{lemma}
\begin{proposition}
For any fixed $\delta >0$, let $N\geq N_{0}\triangleq  8mp+16\log{\frac{1}{\delta}}$. We have
\begin{equation*}
   \mathcal{T}_{p}^{0}E_{-}E_{-}^{*}\mathcal{T}_{p}^{0*}\succeq \frac{N}{4}\sigma_{E}I_{mp}
   \end{equation*}
\label{proposition:lower bound noise}
 with probability at least $1-\delta$.
\end{proposition}
\begin{proof}
For any output trajectory $i$, note that the past innovation $E_{-}^{i}$ has the same distribution as a single Gaussian random vector, $E_{-}^{i}\sim \mathcal{N}(0,\diag(\bar{R}_{0},\cdots,\bar{R}_{p-1}))$, since the innovations $e_{0}^{i},\cdots,e_{p-1}^{i}$ are independent zero-mean Gaussian random vectors with covariance matrices $\Bar{R}_{0},\cdots,\Bar{R}_{p-1}$, respectively. We can rewrite $E_{-}^{i}$ as 
\begin{equation*}
\begin{aligned}
E_{-}^{i}=\diag(\bar{R}_{0}^{\frac{1}{2}},\cdots,\bar{R}_{p-1}^{\frac{1}{2}})\mathbf{u}_{i},
\end{aligned}
\end{equation*}
where $\mathbf{u}_{i}$ are i.i.d random vectors with $\mathbf{u}_{i}\sim \mathcal{N}(0,I_{mp})$. Let $U_{-}= 
\begin{bmatrix}
\mathbf{u}_{1}&\mathbf{u}_{2}&\cdots&\mathbf{u}_{N}\\
\end{bmatrix}$.
Fixing $\delta>0$ and applying Lemma \ref{lemma:Lower bound of unit variance gaussian}, with probability of at least $1-\delta$, we have 
\begin{equation} \label{lam_min}
\begin{aligned}
   \sqrt{\lambda_{min}(U_{-}U_{-}^{*})}\geq \sqrt{N}-\sqrt{mp}-\sqrt{2\log{\frac{1}{\delta}}}.
   \end{aligned}
\end{equation}

Considering the inequality $2(a^2+b^2)\geq(a+b)^{2}$ and the assumption that $N\geq N_{0}\triangleq  8mp+16\log{\frac{1}{\delta}}$, we have
\begin{equation*}
\begin{aligned}
 2(mp+2\log{\frac{1}{\delta}})\geq  (\sqrt{mp}+\sqrt{2\log{\frac{1}{\delta}}})^2,
   \end{aligned}
\end{equation*}
\begin{equation*}
\begin{aligned}
\Rightarrow &\frac{\sqrt{N}}{2}\geq \sqrt{mp}+\sqrt{2\log{\frac{1}{\delta}}}.
\end{aligned}
\end{equation*}

In conjunction with \eqref{lam_min}, this yields $\sqrt{\lambda_{min}(U_{-}U_{-}^{*})} \geq\frac{1}{2}\sqrt{N}$ with probability at least $1-\delta$. Hence, we have
\begin{equation*}
\begin{aligned}
   U_{-}U_{-}^{*}\succeq\frac{N}{4}I_{mp}
   \end{aligned}
\end{equation*}
with probability of at least $1-\delta$.

Finally, multiplying by $\mathcal{T}_{p}^{0} \diag(\bar{R}_{0}^{\frac{1}{2}},\cdots,\bar{R}_{p-1}^{\frac{1}{2}})$ from the left and $\diag(\bar{R}_{0}^{\frac{1}{2}},\cdots,\bar{R}_{p-1}^{\frac{1}{2}})\mathcal{T}_{p}^{0*}$ from the right gives $ \mathcal{T}_{p}^{0}E_{-}E_{-}^{*}\mathcal{T}_{p}^{0*}\succeq\frac{N}{4}\Sigma_{E}\succeq\frac{N}{4}\sigma_{E}I_{mp}$ with probability at least $1-\delta$.
\end{proof}

To bound the cross terms due to the possibly non-zero batch matrix of initial states in \eqref{YY^{*}}, $\hat{X}_{0}E_{-}^{*}$, we will be using the following lemma from \cite[Lemma~A.1]{oymak2018non}.
\begin{lemma} \label{lemma:cross-term}
Let $M\in\mathbb{R}^{m\times n}$ be a matrix with $m\geq n$, and let $\eta\in\mathbb{R}$ be such that $\|M\|\leq \eta$. Let $Z\in\mathbb{R}^{m\times k}$ be a matrix with independent standard normal entries. Then, for all $t\geq0$, with probability at least $1-2\exp(\frac{-t^{2}}{2})$, 
\begin{equation*}
  \|M^{*}Z\|\leq \eta (\sqrt{2(n+k)}+t).
\end{equation*}
\end{lemma}

\begin{proposition}
For any fixed $ \delta> 0$, let $\gamma_{p}\triangleq \mathcal{\bar{R}}_{T}^{\frac{1}{2}} (\sqrt{2(n+mp)}+\sqrt{2\log{\frac{2}{\delta}}})$, and $\gamma_{f}\triangleq \mathcal{\bar{R}}_{T}^{\frac{1}{2}} (\sqrt{2(n+mf)}+\sqrt{2\log{\frac{2}{\delta}}})$. For $N\geq n$, each of these inequalities hold with probability at least $1-\delta$:
\begin{equation*}
\|\hat{X}_{0}E_{-}^{*}\|\leq \|\hat{X}_{0}\|\gamma_{p},
\end{equation*}
\begin{equation*}
\|\hat{X}_{0}E_{+}^{*}\|\leq \|\hat{X}_{0}\|\gamma_{f}.
\end{equation*}
\label{proposition:bound XE}
\end{proposition}
\begin{proof}
We will only show the first inequality as the second one is almost identical. We can rewrite $E_{-}=\diag(\bar{R}_{0}^{\frac{1}{2}},\cdots,\bar{R}_{p-1}^{\frac{1}{2}})U_{-}$, where $U_{-}=
\begin{bmatrix}
\mathbf{u}_{1}&\mathbf{u}_{2}&\cdots&\mathbf{u}_{N}\\
\end{bmatrix}$, where $\mathbf{u}_{i}$ are i.i.d random vectors with $\mathbf{u}_{i}\sim \mathcal{N}(0,I_{mp})$.
Applying Lemma \ref{lemma:cross-term}, we obtain
\begin{equation*}
\begin{aligned}
\|\hat{X}_{0}E_{-}^{*}\|&=\|\hat{X}_{0}U_{-}^{*}\diag(\bar{R}_{0}^{\frac{1}{2}},\cdots,\bar{R}_{p-1}^{\frac{1}{2}})\|\\
&\leq \|\hat{X}_{0}\|\mathcal{\bar{R}}_{T}^{\frac{1}{2}}(\sqrt{2(n+mp)}+t)
\end{aligned}
\end{equation*}
with probability at least $1-2\exp(\frac{-t^{2}}{2})$.
Finally, setting $\delta=2\exp(\frac{-t^{2}}{2})$, we have $t=\sqrt{2\log{\frac{2}{\delta}}}$. Plugging $t$ into the above inequality, we get the desired form.  
\end{proof}

Now we are ready to show that $\|(Y_{-}Y_{-}^{*})^{-1}\|$ is decreasing with a rate of $\mathcal{O}({\frac{1}{N}})$.
\begin{lemma}
Fix any $\delta >0$ and let $N \geq \max\{N_0,N_1\}$, where $N_{0}= 8mp+16\log{\frac{1}{\delta}}$, and $N_1 \triangleq\frac{16\|\mathcal{T}_{p}^{0}\|\|\mathcal{O}_{p}\|\|\hat{X}_{0}\|\gamma_{p}}{\sigma_{E}}$. Define $\zeta \triangleq \mathcal{O}_{p}\mu\mu^{*}\mathcal{O}_{p}^{*}$. We have
\begin{equation*}
  \|(Y_{-}Y_{-}^{*})^{-1}\|\leq \frac{8}{N\sigma_{min}(\sigma_{E}I_{mp}+8\zeta)}
\end{equation*}
\label{lemma:Output PE}
with probability at least $1-2\delta$.
\end{lemma}
\begin{proof}
Recall the explicit form of $Y_{-}Y_{-}^{*}$ in \eqref{YY^{*}}. Letting $u\in \mathbb{R}^{mp}$ be an arbitrary unit vector, we can write
\begin{equation*}
\begin{aligned}
u^{*}Y_{-}Y_{-}^{*}u&=u^{*}\mathcal{O}_{p}\hat{X}_{0}\hat{X}_{0}^{*}\mathcal{O}_{p}^{*}u+u^{*}\mathcal{T}_{p}^{0}E_{-}E_{-}^{*}\mathcal{T}_{p}^{0*}u\\
&+u^{*}\mathcal{O}_{p}\hat{X}_{0}E_{-}^{*}\mathcal{T}_{p}^{0*}u+u^{*}\mathcal{T}_{p}^{0}E_{-}\hat{X}_{0}^{*}\mathcal{O}_{p}^{*}u\\
&\geq u^{*}\mathcal{O}_{p}\hat{X}_{0}\hat{X}_{0}^{*}\mathcal{O}_{p}^{*}u+u^{*}\mathcal{T}_{p}^{0}E_{-}E_{-}^{*}\mathcal{T}_{p}^{0*}u\\
&-2\|\hat{X}_{0}E_{-}^{*}\|\|\mathcal{T}_{p}^{0}\|\|\mathcal{O}_{p}\|, 
\end{aligned}
\end{equation*}
where we used the Cauchy–Schwarz inequality. Fixing $ \delta >0$, letting $N\geq N_{0}$, applying Proposition \ref{proposition:lower bound noise} and Proposition \ref{proposition:bound XE}, and using a union bound,  we have 
\begin{equation*}
\begin{aligned}
u^{*}Y_{-}Y_{-}^{*}u&\geq u^{*}\mathcal{O}_{p}\hat{X}_{0}\hat{X}_{0}^{*}\mathcal{O}_{p}^{*}u+\frac{N}{4}\sigma_{E}-2\|\mathcal{T}_{p}^{0}\|\|\mathcal{O}_{p}\|\|\hat{X}_{0}\|\gamma_{p} \label{uYYu}
\end{aligned}
\end{equation*}
with probability at least $1-2\delta$.

Conditioning on the above event and letting $N\geq N_1=\frac{16\|\mathcal{T}_{p}^{0}\|\|\mathcal{O}_{p}\|\|\hat{X}_{0}\|\gamma_{p}}{\sigma_{E}}$, we have 
\begin{equation*}
\begin{aligned} 
u^{*}Y_{-}Y_{-}^{*}u&\geq  u^{*}\mathcal{O}_{p}\hat{X}_{0}\hat{X}_{0}^{*}\mathcal{O}_{p}^{*}u+\frac{N}{8}\sigma_{E}\\&=u^{*}\mathcal{O}_{p}N\mu\mu^{*}\mathcal{O}_{p}^{*}u+u^{*}\frac{N}{8}\sigma_{E}I_{mp}u,
\end{aligned}
\end{equation*}
where the equality is due to the fact that $\hat{x}^{i}_{0}=\mu$ for all $i$.

Consequently, we have
\begin{equation*}
\begin{aligned} \
Y_{-}Y_{-}^{*}&\succeq \mathcal{O}_{p}N\mu\mu^{*}\mathcal{O}_{p}^{*}+\frac{N}{8}\sigma_{E}I_{mp}.
\end{aligned}
\end{equation*}

Taking the inverse we get the desired result.
\end{proof}

To see that $N$ will eventually be greater than $N_{1}$ even if $\|\hat{X}_{0}\|$ is non-zero, note that when the system is strictly stable or marginally stable, $\|\mathcal{T}_{p}^{0}\|$ and $\|\mathcal{O}_{p}\|$ will grow no faster than $\mathcal{O}(p^\mathbf{d})$ for some constant $\mathbf{d}$ (see Proposition \ref{proposition:Bound Block-Toepliz Matrix} for $\|\mathcal{T}_{p}^{0}\|$, and \cite[Corollary~E.1]{tsiamis2019finite} for $\|\mathcal{O}_{p}\|$). Further, $\gamma_{p}$ is $\mathcal{O}(p^{\frac{1}{2}})$, and $\|\hat{X}_{0}\|=\sqrt{N}\|\mu\|=\mathcal{O}(\sqrt{N})$. Thus if $p=\mathcal{O}(\log{N})$, $N$ will eventually be greater than $N_{1}$ as $N$ increases.

Now we will show that the term $\|E_{+}E_{-}^{*}\|$ is $\mathcal{O}(\sqrt{N})$. We will leverage the following lemma from\cite[Lemma~1]{dean2019sample} to bound the product of the innovation terms.
\begin{lemma}
Let $f_{i}\in\mathbb{R}^{m}$, $g_{i}\in\mathbb{R}^{n}$ be independent random vectors $f_{i}\sim \mathcal{N}(0,\Sigma_{f})$ and $g_{i}\sim \mathcal{N}(0,\Sigma_{g})$, for $i=1,\cdots,N$. Let $N\geq2(n+m)\log{\frac{1}{\delta}}$. For any fixed $\delta >0$, we have
\begin{equation*}
   \|\sum_{i=1}^{N}f_{i}g_{i}^{*}\|\leq4\|\Sigma_{f}\|^{\frac{1}{2}}\|\Sigma_{g}\|^{\frac{1}{2}}\sqrt{N(m+n)\log{\frac{9}{\delta}}}.
\end{equation*}
\label{lemma:upper bound two independent gaussian}
with probability at least $1-\delta$.
\end{lemma}
\begin{proposition}\label{proposition:Upper bound noise}
For any fixed $ \delta >0$, let $N\geq N_{2}\triangleq 2(mf+mp)\log{\frac{1}{\delta}}$. We have
\begin{equation*}
\begin{aligned}
   &\|\mathcal{T}_{f}^{p}E_{+}E_{-}^{*}\mathcal{T}_{p}^{0*}\|\leq 4\|\mathcal{T}_f^{p}\|\|\mathcal{T}_p^{0}\|\mathcal{\bar{R}}_{T}\sqrt{N(mf+mp)\log{\frac{9}{\delta}}}
\end{aligned}
   \end{equation*}
 with probability at least $1-\delta$.
\end{proposition}
\begin{proof}
Note that the columns of $E_{+}$ are independent Gaussian random vectors, i.e., $E_{+}^{i}\sim \mathcal{N}(0,\diag(\bar{R}_{p},\cdots,\bar{R}_{p+f-1}))$. Similarly, the columns of $E_{-}$ are independent Gaussian random vectors, i.e., $E_{-}^{i}\sim \mathcal{N}(0,\diag(\bar{R}_{0},\cdots,\bar{R}_{p-1}))$. Further, $E_{+}^{i}$ and $E_{-}^{i}$ are independent from classical results of Kalman filtering theory \cite{anderson2012optimal}. Applying Lemma \ref{lemma:upper bound two independent gaussian}, we get the desired result.
\end{proof}

With Proposition \ref{lemma:upper bound two independent gaussian}, we are now in place to prove the bound on the estimation error of $G$ .

\begin{theorem}[Bound on estimation error of $G$] \label{thm:Bound for $G$ general}
Consider the Kalman filter form \eqref{KalmanFilterSystem} of system \eqref{originalStateSpace} under Assumptions \ref{asm: standard} and \ref{asm: controllable}, and let ${G}$ be defined as in \eqref{True G}. For any fixed $\delta >0$, let $\hat{G}$ defined in \eqref{Ghat} be the output of the linear regression described in Algorithm \ref{Linear Regression} given $N$ trajectories of outputs, where $N\geq \max\{N_0,N_1,N_2\}$, where $N_{0}= 8mp+16\log{\frac{1}{\delta}},N_{1}=\frac{16\|\mathcal{T}_{p}^{0}\|\|\mathcal{O}_{p}\|\|\hat{X}_{0}\|\gamma_{p}}{\sigma_{E}},N_{2}=2(mf+mp)\log{\frac{1}{\delta}}$. We have:
\begin{equation*}
\begin{aligned}
\|\hat{G}-G\|\leq \frac{\epsilon_{1}}{\sqrt{N}\sigma_{min}(\sigma_{E}I_{mp}+8\zeta)}+\frac{\|\hat{X}_{0}\|\epsilon_{2}+\|\hat{X}_{0}\|^2\epsilon_{3}}{N\sigma_{min}(\sigma_{E}I_{mp}+8\zeta)}
\end{aligned}
\end{equation*}
with probability at least $1-4\delta$, where 
\begin{equation*}
\begin{aligned}
&\epsilon_{1}=32\|\mathcal{T}_f^{p}\|\|\mathcal{T}_p^{0}\|\mathcal{\bar{R}}_{T}\sqrt{(mf+mp)\log{\frac{9}{\delta}}}\\
&\epsilon_{2}=8\gamma_{f}\|\mathcal{T}_{f}^{p}\|\|\mathcal{O}_{p}\|\\
&+8\|(A-K_{p-1}C)\cdots(A-K_{0}C)\|\gamma_{p}\|\mathcal{T}_{p}^{0}\|\|\mathcal{O}_{f}\|,\\
&\epsilon_{3}=8\|\mathcal{O}_{f}\|\|\mathcal{O}_{p}\|\|(A-K_{p-1}C)\cdots(A-K_{0}C)\|,\\
&\gamma_{p}= \mathcal{\bar{R}}_{T}^{\frac{1}{2}} (\sqrt{2(n+mp)}+\sqrt{{2\log{\frac{2}{\delta}}}}),\\ &\gamma_{f}= \mathcal{\bar{R}}_{T}^{\frac{1}{2}} (\sqrt{2(n+mf)}+\sqrt{{2\log{\frac{2}{\delta}}}}), \quad \zeta \triangleq \mathcal{O}_{p}\mu\mu^{*}\mathcal{O}_{p}^{*}.\\
\end{aligned}
\end{equation*}
\end{theorem}
\begin{proof}
Recall the expression of the error $\hat{G}-G$ in \eqref{error_G}. We first bound the error term $\mathcal{T}_{f}^{p}E_{+}Y_{-}^{*}(Y_{-}Y_{-}^{*})^{-1}$. Using $Y_{-}=\mathcal{O}_{p}\hat{X}_{0}+\mathcal{T}_{p}^{0}E_{-}$, we have 
\begin{equation*}
\begin{aligned}
\|\mathcal{T}_{f}^{p}E_{+}Y_{-}^{*}(Y_{-}Y_{-}^{*})^{-1}\|&\leq\|\mathcal{T}_{f}^{p}E_{+}\hat{X}_{0}^{*}\mathcal{O}_{p}^{*}\|\|(Y_{-}Y_{-}^{*})^{-1}\|\\
&+\|\mathcal{T}_{f}^{p}E_{+}E_{-}^{*}\mathcal{T}_{p}^{0*}\|\|(Y_{-}Y_{-}^{*})^{-1}\|.
\end{aligned}
\end{equation*}

Fix $ \delta >0$ and let $N\geq \max\{N_0,N_1,N_2\}$. Applying Proposition \ref{proposition:bound XE}, Proposition \ref{proposition:Upper bound noise} and Lemma \ref{lemma:Output PE} to the above inequality and using a union bound, we obtain
\begin{equation}
\begin{aligned}
&\|\mathcal{T}_{f}^{p}E_{+}Y_{-}^{*}(Y_{-}Y_{-}^{*})^{-1}\|\leq\frac{8\|\hat{X}_{0}\|\gamma_{f}\|\mathcal{T}_{f}^{p}\|\|\mathcal{O}_{p}\|}{N\sigma_{min}(\sigma_{E}I_{mp}+8\zeta)}\\
&+\frac{32\|\mathcal{T}_f^{p}\|\|\mathcal{T}_p^{0}\|\mathcal{\bar{R}}_{T}\sqrt{(mf+mp)\log{\frac{9}{\delta}}}}{\sqrt{N}\sigma_{min}(\sigma_{E}I_{mp}+8\zeta)} \label{error bound 1}
\end{aligned}
\end{equation}
with probability at least $1-4\delta$. 

Second, we bound the error term $\mathcal{O}_{f}(A-K_{p-1}C)\cdots(A-K_{0}C)\hat{X}_{0}Y_{-}^{*}(Y_{-}Y_{-}^{*})^{-1}$. We have
\begin{equation*}
\begin{aligned}
\|\hat{X}_{0}Y_{-}^{*}\|=\|\hat{X}_{0}\hat{X}_{0}^{*}\mathcal{O}_{p}^{*}+\hat{X}_{0}E_{-}^{*}\mathcal{T}_{p}^{0*}\|.
\end{aligned}
\end{equation*} 
Conditioning on the event $\|\hat{X}_{0}E_{-}^{*}\|\leq\|\hat{X}_{0}\|\gamma_{p}$ from Proposition \ref{proposition:bound XE}, we have
\begin{equation*}
\begin{aligned}
\|\hat{X}_{0}Y_{-}^{*}\|\leq \|\hat{X}_{0}\|^{2}\|\mathcal{O}_{p}\|+\|\hat{X}_{0}\|\gamma_{p}\|\mathcal{T}_{p}^{0}\|.
\end{aligned}
\end{equation*}
Conditioning on the above event and the event $\|(Y_{-}Y_{-}^{*})^{-1}\|\leq \frac{8}{N\sigma_{min}(\sigma_{E}I_{mp}+8\zeta)}$ from Lemma \ref{lemma:Output PE}, we have
\begin{equation}
\begin{aligned}
&\|\mathcal{O}_{f}(A-K_{p-1}C)\cdots(A-K_{0}C)\hat{X}_{0}Y_{-}^{*}(Y_{-}Y_{-}^{*})^{-1}\|\leq\\
&\|\mathcal{O}_{f}\|\|(A-K_{p-1}C)\cdots(A-K_{0}C)\|\|\hat{X}_{0}Y_{-}^{*}\|\|(Y_{-}Y_{-}^{*})^{-1}\|\\
&\leq\frac{8\|\mathcal{O}_{f}\|\|(A-K_{p-1}C)\cdots(A-K_{0}C)\|\|\hat{X}_{0}\|^{2}\|\mathcal{O}_{p}\|}{N\sigma_{min}(\sigma_{E}I_{mp}+8\zeta)}\\
&+\frac{8\|\mathcal{O}_{f}\|\|(A-K_{p-1}C)\cdots(A-K_{0}C)\|\|\hat{X}_{0}\|\gamma_{p}\|\mathcal{T}_{p}^{0}\|}{N\sigma_{min}(\sigma_{E}I_{mp}+8\zeta)}.\label{error bound 2}
\end{aligned}
\end{equation}

Finally, we combine the two upper bounds from \eqref{error bound 1} and \eqref{error bound 2} to get the desired form.
\end{proof}

\textbf{Interpretation of Theorem \ref{thm:Bound for $G$ general}.}

\textit{Learning rate when $\|\hat{X}_{0}\|$ is zero, and the effects of trajectory length:} When $\|\hat{X}_{0}\|=0$, i.e., the initial state of the system \eqref{originalStateSpace} has zero mean, the upper bound of the error will not depend on $\epsilon_{2},\epsilon_{3}$. Noting the dependencies on $p,f$ in $\epsilon_{1}$, setting $p$ and $f$ to be small will generally result in a smaller error bound of $G$, since we are estimating a smaller $G$. However, $p,f$ should be greater than the order $n$ (and $p$ should also be large enough such that Assumption \ref{asm: controllable} is satisfied), so that Algorithm \ref{alg: SVD Algorithm} can recover the system matrices from $G$. The estimator $\hat{G}$ can achieve a learning rate of $\mathcal{O}(\frac{1}{\sqrt{N}})$. This rate is faster than the single trajectory case reported in \cite{tsiamis2019finite} in that there are no logarithmic factors, and it applies to both stable and unstable systems. This confirms the benefits of being able to collect multiple independent trajectories starting from $x_{0}\sim \mathcal{N}(0,\Sigma_{0})$.

\textit{Learning rate when $\|\hat{X}_{0}\|$ is nonzero, and the effects of trajectory length:} When $\|\hat{X}_{0}\|$ is  nonzero, the error bound will depend on $\epsilon_{2},\epsilon_{3}$. Note that $\|\hat{X}_{0}\|=\sqrt{N}\|\mu\|$ when the initial state of each trajectory has mean $\mu$. The term $\frac{\|\hat{X}_{0}\|^2\epsilon_{3}}{N\sigma_{min}(\sigma_{E}I_{mp}+8\zeta)}$ is $\mathcal{O}(1)$ when $p$ is fixed. In such case, if the system is known to be marginally stable ($\rho(A)\leq 1$), we can leverage the fact that the norm $\|(A-K_{p-1}C)\cdots(A-K_{0}C)\|$ in $\epsilon_{3}$ converges to zero exponentially fast with $p$ (see Proposition \ref{proposition:Convergence of (A-KC)} in the Appendix), by setting $p=c\log{N}$ for some sufficiently large $c$, to force the term $\|(A-K_{p-1}C)\cdots(A-K_{0}C)\|$ to go to zero no slower than $\mathcal{O}(\frac{1}{\sqrt{N}})$. The term $\mathcal{\bar{R}}_{T}$ is $\mathcal{O}(1)$ since the Kalman filter converges \cite{anderson2012optimal}. For the same reason, by fixing a small $f>n$, $\|\mathcal{T}_{f}^{p}\|$ is $\mathcal{O}(1)$. In addition, $\|\mathcal{T}_{p}^{0}\|$ and $\|\mathcal{O}_{p}\|$ are $\mathcal{O}(1)$ for stable systems, and $\mathcal{O}(p^\mathbf{d})$ for some constant $\mathbf{d}$ for marginally stable systems (see Proposition \ref{proposition:Bound Block-Toepliz Matrix} in the Appendix for $\|\mathcal{T}_{p}^{0}\|$, and \cite[Corollary~E.1]{tsiamis2019finite} for $\|\mathcal{O}_{p}\|$). As a result, the error will decrease with a rate of $\mathcal{O}(\sqrt{\frac{\log{N}}{N}})$ for strictly stable systems, and $\mathcal{O}(\frac{(\log{N})^d}{\sqrt{N}})$ for some constant $d$ for marginally stable systems, even if $\|\hat{X}_{0}\|$ is non-zero.

The next step shows that the realization error of system matrix estimates $(\hat{A},\hat{C},\hat{K}_{p-1})$ provided by Algorithm \ref{alg: SVD Algorithm} is bounded. Based on our assumption that $\mathcal{O}_{f}$ and $\mathcal{K}_{p}$ have rank $n$, the true $G$ also has rank $n$. The proof of the following theorem entirely follows \cite[Theorem ~4]{tsiamis2019finite}, with the only difference being the replacement of steady state Kalman gain $K$ by non-steady state Kalman gain $K_{p-1}$.
\begin{theorem}
[Bound on realizations of system matrices]
Let $G$ and $\hat{G}$ be defined in \eqref{True G} and \eqref{Ghat}. Let the estimates based on $\hat{G}$ using Algorithm \ref{alg: SVD Algorithm} be $\hat{\mathcal{O}}_f,\hat{\mathcal{K}}_p,\hat{A},\hat{C},\hat{K}_{p-1}$, and the corresponding matrices based on the true $G$ using Algorithm \ref{alg: SVD Algorithm} be $\tilde{\mathcal{O}}_f,\tilde{\mathcal{K}}_p,\tilde{A},\tilde{C},\tilde{K}_{p-1}$. If $G$ has rank $n$ and
$
   \|\hat{G}-G\|\leq \frac{\sigma_{n}(G)}{4},
$
then there exists an orthonormal matrix $\mathcal{T}\in\mathbb{R}^{n\times n}$ such that:
\begin{equation*}
\begin{aligned}
&\|\hat{\mathcal{O}}_f-\tilde{\mathcal{O}}_{f}\mathcal{T}\|\leq 2\sqrt{\frac{10n}{\sigma_{n}(G)}}\|\hat{G}-G\|,\\
&\|\hat{C}-\tilde{C}\mathcal{T}\|\leq \|\hat{\mathcal{O}}_f-\tilde{\mathcal{O}}_{f}\mathcal{T}\|,\\
&\|\hat{A}-\mathcal{T}^{*}\tilde{A}\mathcal{T}\|\leq \frac{\sqrt{\|G\|}+\sigma_{o}}{\sigma_{o}^{2}}\|\hat{\mathcal{O}}_f-\tilde{\mathcal{O}}_{f}\mathcal{T}\|,\\
&\|\hat{K}_{p-1}-\mathcal{T}^{*}\tilde{K}_{p-1}\|\leq 2\sqrt{\frac{10n}{\sigma_{n}(G)}}\|\hat{G}-G\|,\\
\end{aligned}
\end{equation*}
where $\sigma_{o}\triangleq min(\sigma_n(\hat{\mathcal{O}}_{f}^{u}),\sigma_n(\tilde{\mathcal{O}}_{f}^{u}))$\label{thm:Bound of realization}. Recall that the notation $\hat{\mathcal{O}}_{f}^{u},\tilde{\mathcal{O}}_{f}^{u}$ refers to the submatrix formed by the top $m(f-1)$ rows of the respective matrix. 
\end{theorem}

\begin{Remark}
Note that the matrices $\tilde{A},\tilde{C}, \tilde{K}_{p-1}$ are equivalent to the original $A,C,K_{p-1}$ matrices up to a similarity transformation. As $p$ increases, $\|G\|=\|\mathcal{O}_{f}\mathcal{K}_{p}\|$ is $\mathcal{O}(1)$ since $f$ is fixed, and $\|\mathcal{K}_{p}\|$ is also $\mathcal{O}(1)$ (see Proposition \ref{proposition:Bound Controllability Matrix} in the Appendix). From Proposition \ref{proposition:Lower bound of least singular value of G} in the Appendix, $\sigma_{n}(G)$ is lower bounded as $p$ increases. As suggested in\cite[Remark ~3]{tsiamis2019finite}, the random term $\sigma_n(\hat{\mathcal{O}}_{f}^{u})$ in $\sigma_{o}$ can be replaced by a deterministic bound as 
\begin{equation*}
\begin{aligned}
&\sigma_n(\hat{\mathcal{O}}_{f}^{u})\geq\sigma_n(\tilde{\mathcal{O}}_{f}^{u})-\|\hat{\mathcal{O}}_f-\tilde{\mathcal{O}}_{f}\mathcal{T}\|.
\end{aligned}
\end{equation*}
Hence $\sigma_{o}$ will be lower bounded by $\frac{\sigma_n(\tilde{\mathcal{O}}_{f}^{u})}{2}>0$ when the error $\|\hat{\mathcal{O}}_f-\tilde{\mathcal{O}}_{f}\mathcal{T}\|$ is small enough, where the inequality is due to the fact that we assumed the system is observable. Consequently, the term $\frac{\sqrt{\|G\|}+\sigma_{o}}{\sigma_{o}^{2}}$ is always $\mathcal{O}(1)$.

As a result, all estimation errors of system matrices depend linearly on $\|\hat{G}-G\|$, even if $p$ is increasing. Hence, the realization error will decrease at least as fast as $\mathcal{O}(\frac{1}{\sqrt{N}})$ when $\|\hat{X}_{0}\|=0$, and $p$ is fixed. When $\|\hat{X}_{0}\|$ is non-zero, the error can decrease at a rate of $\mathcal{O}(\sqrt{\frac{\log{N}}{N}})$ for strictly stable systems, and at a rate of  $\mathcal{O}(\frac{(\log{N})^d}{\sqrt{N}})$ for some constant $d$ for marginally stable systems by setting $p=c\log{N}$ for some positive constant $c$.  Note that as $p$ goes to infinity, the matrix $\hat{K}_{p-1}$ estimates the steady state Kalman gain $\mathcal{T}^{*}\tilde{K}$. 

On the other hand, the dependencies on $\sigma_{n}(G)$ and $\sigma_n(\hat{\mathcal{O}}_{f}^{u})$ also show that the estimation error of system matrices depends on the “normalized estimation error” of $G$. Consequently, although our bound suggests that setting $p,f$ to be small could potentially reduce the estimation error of $G$ (when $\mu=0$), it may not necessarily reduce the error of the system matrices. A similar issue also appears in the recovery of system matrices from Markov parameters \cite{zheng2020non}. It is of interest to study how trajectory length directly affects the realization error in future work.
\end{Remark}



\section{Conclusion and future work} \label{sec: conclusion}
In this paper, we performed finite sample analysis of learning the dynamics of autonomous systems using multiple trajectories. Our results rely neither on controlled inputs, nor on observations of steady state behaviors of the system. We proved a learning rate that is consistent with \cite{zheng2020non} and \cite{tsiamis2019finite} (up to logarithmic factors). Future work could focus on understanding how to effectively utilize multiple trajectories of varying lengths, and how other variants of the balanced realization algorithm affect the error.



\appendix

\section{Appendix}

\subsection{Auxiliary results}
Some of the results here are standard, and we include them for completeness.

\begin{lemma}(\cite[Lemma~E.2]{tsiamis2019finite}). \label{lemma:Bound of power of A}
Consider the series $S_{t}=\sum_{i=0}^{t}\|A^{i}\|$. If the matrix $A$ is strictly stable ($\rho(A)<1$), then $S_{t}=\mathcal{O}(1)$;
if the matrix $A$ is marginally stable ($\rho(A)=1$), then $S_{t}=\mathcal{O}(t^\mathbf{d})$, where $\mathbf{d}$ is the largest Jordan block of $A$ corresponding to a unit circle eigenvalue $\|\lambda\|=1$.

\end{lemma}


\begin{proposition}
The norm $\|\mathcal{T}_{p}^{0}\|$ is $\mathcal{O}(p^{\mathbf{d}})$ with $p$ for some constant $\mathbf{d}$ when the system matrix $A$ is marginally stable, and is $\mathcal{O}(1)$ when the system matrix $A$ is strictly stable. \label{proposition:Bound Block-Toepliz Matrix}
\end{proposition}
\begin{proof}
Letting $K_{max}(p)=\max_{t\in0,\dots,p-2}\|K_{t}\|$, where $K_{t}$ is defined in \eqref{gain}. We have
\begin{equation*}
\begin{aligned}
&\|\mathcal{T}_{p}^{0}\|\leq\|I_{mp}\|+\|C\|K_{max}(p)+\\
&\|C\|\|A\|K_{max}(p)+\cdots+\|C\|\|A^{p-2}\|K_{max}(p)\\
&=1+\|C\|K_{max}(p)\sum_{i=0}^{p-2}\|A^{i}\|.
\end{aligned}
\end{equation*}

From Kalman filtering theory, the Kalman gain $K_{t}$ converges to its steady state $K$ under Assumption \ref{asm: standard}. Hence $K_{max}(p)=\mathcal{O}(1)$. From Lemma \ref{lemma:Bound of power of A}, we have the above sum is $\mathcal{O}(1)$ if $A$ is strictly stable, and $\mathcal{O}(p^\mathbf{d})$ when $A$ is marginally stable.
\end{proof}

\begin{lemma}(\cite[Theorem~6.6]{hartfiel2002nonhomogeneous}).\label{lemma:Pertubation Bound} Let $U+A_{1},U+A_{2},\cdots$ be a sequence of $n \times n$ matrices. Given $\epsilon>0$, there is a $\delta(\epsilon)$ such that if $\|A_{k}\|\leq\delta(\epsilon)$ for all $k$, then   
\begin{align*}   
\|(U+A_{k})\cdots(U+A_{1})\|\leq \sigma(\rho(U)+\epsilon)^{k}
\end{align*}
for some constant $\sigma$.
\end{lemma}

\begin{proposition} \label{proposition:Convergence of (A-KC)}
For any fixed integer $k$, where $p-1\geq k\geq 0$, we have 
\begin{align*}
\|(A-K_{p-1}C)\cdots(A-K_{k}C)\|=\mathcal{O}(e^{-c_{0}p}),
\end{align*}
for some positive constant $c_{0}$.
\end{proposition}

\begin{proof}
From Kalman filtering theory, the Kalman gain $K_{t}$ converges to its steady state $K$, and the matrix $A-KC$ has spectral radius less than 1 under Assumption \ref{asm: standard}. Hence, we can write $A-K_{t}C=A-KC+\eta_{t}$ for $t\geq0$, where $\|\eta_{t}\|$ converges to $0$. Pick $\epsilon$ such that $\epsilon+\rho(A-KC)< 1$. To apply Lemma \ref{lemma:Pertubation Bound}, let $U=A-KC$, and let $k+t(\epsilon)$ be the smallest index such that $\|\eta_{t}\|\leq \delta(\epsilon)$ for all $t\geq k+t(\epsilon)$. Letting $p\geq k+t(\epsilon)+1$, we have
\begin{equation*}
\begin{aligned}
&\|(A-K_{p-1}C)\cdots(A-K_{k}C)\|\leq\\
&\|\prod_{t=1}^{p-k-t(\epsilon)}(U+\eta_{p-t})\|\|\prod_{t=p-k-t(\epsilon)+1}^{p-k}(U+\eta_{p-t})\|\\
&\leq \sigma(\rho(U)+\epsilon)^{p-k-t(\epsilon)}\|\prod_{t=p-k-t(\epsilon)+1}^{p-k}(U+\eta_{p-t})\|\\
&=\mathcal{O}(e^{-c_{0}p}),
\end{aligned}
\end{equation*}
where the second inequality comes from Lemma \ref{lemma:Pertubation Bound}.
\end{proof}

\begin{proposition}
The norm $\|\mathcal{K}_{p}\|$ is $\mathcal{O}(1)$ with $p$. \label{proposition:Bound Controllability Matrix}
\end{proposition}
\begin{proof}
From Kalman filtering theory, the Kalman gain $K_{p-1}$ converges to its steady state $K$, and the matrix $A-KC$ has spectral radius less than 1 under Assumption \ref{asm: standard}. Hence for any $t\geq1$ we can write $A-K_{t}C=A-KC+\eta_{t}$, where $\|\eta_{t}\|$ converges to $0$. Let $K_{max}(p)=\max_{t\in0,\dots,p-1}\|K_{t}\|$. We have $K_{max}(p)=\mathcal{O}(1)$ since $K_{t}$ converges to $K$. Pick $\epsilon$ such that $\epsilon+\rho(A-KC)< 1$. To apply Lemma \ref{lemma:Pertubation Bound}, let $U=A-KC$, and let $t(\epsilon)$ be the smallest index such that $\|\eta_{t}\|\leq \delta(\epsilon)$ for all $t\geq t(\epsilon)$. Letting $p\geq t(\epsilon)+1$, we have
\begin{equation*} 
\begin{aligned}
&\|\mathcal{K}_{p}\|\leq K_{max}(p)+K_{max}(p)\sum_{t=2}^{p}\|\prod_{j=2}^{t} (U+\eta_{p-j+1})\|\\
&= K_{max}(p)+K_{max}(p)\sum_{t=p-t(\epsilon)+2}^{p}\|\prod_{j=2}^{t} (U+\eta_{p-j+1})\|\\
&+K_{max}(p)\sum_{t=2}^{p-t(\epsilon)+1}\|\prod_{j=2}^{t} (U+\eta_{p-j+1})\|.
\end{aligned}
\end{equation*}

From Proposition \ref{proposition:Convergence of (A-KC)}, we have
\begin{equation} \label{controllability bound 2}
\begin{aligned}
&K_{max}(p)\sum_{t=p-t(\epsilon)+2}^{p}\|\prod_{j=2}^{t} (U+\eta_{p-j+1})\|=\mathcal{O}(e^{-c_{0}p}).
\end{aligned}
\end{equation}

From Lemma \ref{lemma:Pertubation Bound}, we have
\begin{equation} \label{controllability bound 3}
\begin{aligned}
&K_{max}(p)\sum_{t=2}^{p-t(\epsilon)+1}\|\prod_{j=2}^{t} (U+\eta_{p-j+1})\|\leq\\ &K_{max}(p)\sum_{t=2}^{p-t(\epsilon)+1}\sigma(\rho(U)+\epsilon)^{t-1}\leq \frac{K_{max}(p)\sigma}{1-\rho(U)-\epsilon}=\mathcal{O}(1),
\end{aligned}
\end{equation}
where the second inequality comes from geometric series.

Finally, combining \eqref{controllability bound 2} and \eqref{controllability bound 3}, we obtain $\|\mathcal{K}_{p}\|=\mathcal{O}(1)$.

\end{proof}

\begin{proposition}\label{proposition:Lower bound of least singular value of G}
Assume that $\rank(\mathcal{O}_{f})=\rank(\mathbf{K}_{p})=n$, where $n$ is the order of the system. Fix any positive integer $k, n\leq k< p$. Let $\mathbf{K}_{ss}$ be the matrix formed by the last $k$ block columns of the reversed extended controllability matrix $\mathbf{K}_{p}$. For sufficiently large $p$, we have the following inequalities:
\begin{align*}
&\sigma_{n}(G)\geq\frac{\sigma_{n}(\mathcal{O}_{f}\mathbf{K_{ss}})}{2}>0,\\
&\sigma_{n}(\mathcal{K}_{p})\geq\frac{\sigma_{n}(\mathbf{K_{ss}})}{2}>0.
\end{align*}
\end{proposition}
\begin{proof}
We will only show the first inequality as the second one is similar. Recall the definition of $G$ in \eqref{True G}. We can rewrite $G=[M\quad \mathcal{O}_{f}\mathcal{K}_{tv}]$, where $\mathcal{K}_{tv}$ is the matrix formed by the last $k$ block columns of  $\mathcal{K}_{p}$, and $M$ is some residual matrix. We have
\begin{equation*}
\begin{aligned}
GG^{*}\succeq \mathcal{O}_{f}\mathcal{K}_{tv}\mathcal{K}^{*}_{tv}\mathcal{O}_{f}^{*}.
\end{aligned}
\end{equation*}
Hence, we have
\begin{equation*}
\begin{aligned}
\sigma_{n}(G)&\geq\sigma_{n}(\mathcal{O}_{f}\mathcal{K}_{tv})=\sigma_{n}(\mathcal{O}_{f}\mathbf{K}_{ss}+\mathcal{O}_{f}\mathcal{K}_{tv}-\mathcal{O}_{f}\mathbf{K}_{ss})\\
&\geq\sigma_{n}(\mathcal{O}_{f}\mathbf{K}_{ss})-\|\mathcal{O}_{f}\|\|\mathcal{K}_{tv}-\mathbf{K}_{ss}\|,
\end{aligned}
\end{equation*}
where the last inequality comes from the application of \cite[Theorem~3.3.16(c)]{horn1991topics}. From Kalman filtering theory, the Kalman gain $K_{t}$ converges to its steady state $K$ under Assumption \ref{asm: standard}. Consequently, we have $\|\mathcal{K}_{tv}-\mathbf{K}_{ss}\|$ converges to zero as $p$ increases.  We see that $\sigma_{n}(G)$ will eventually be lower bounded by $\frac{\sigma_{n}(\mathcal{O}_{f}\mathbf{K}_{ss})}{2}>0$ as $p$ increases, where the inequality comes from the assumption that $\mathcal{O}_{f}$ and $\mathbf{K}_{p}$ have full rank and the Cayley-Hamilton Theorem. 
\end{proof}

\bibliographystyle{IEEEtran}
\bibliography{main}

\end{document}